\documentclass[letterpaper, 10 pt, conference]{ieeeconf}  
\let\labelindent\relax
\usepackage{comment}
\usepackage{cite} 
\usepackage{enumitem}

\usepackage{amssymb, mathbbol, mathtools, amsthm} 
\usepackage{csquotes} 
\usepackage{hyperref} 
\usepackage{pifont} 
\usepackage{xcolor} 
\usepackage{textcomp}
\usepackage[export]{adjustbox}
\usepackage{setspace}
\usepackage{float}
\IEEEoverridecommandlockouts                            
\usepackage{array} 
\usepackage{graphics} 
\usepackage{epsfig} 
\usepackage{caption}
\usepackage{mathptmx} 
\usepackage{dcolumn}
\usepackage{times} 
\newtheorem{theorem}{Theorem}
\newtheorem{lemma}{Lemma}
\newtheorem{remark}{Remark}
\theoremstyle{definition}
\newtheorem{definition}{Definition}
\usepackage{subcaption}
\theoremstyle{remark}
\newtheorem{example}{Example}

\title{\LARGE \bf Real Eventual Exponential Positivity of Complex-valued Laplacians:Applications to Consensus in Multi-agent Systems}

\author{Aditi Saxena$^{1}$ \hspace{25 mm} Twinkle Tripathy$^{2}$\hspace{25 mm} Rajasekhar Anguluri$^{3}$
\thanks{$^{1}$Aditi Saxena is a research scholar with the Department of Electrical Engineering, Indian Institute of Technology Kanpur, India. (email: {\tt\small saditi23@iitk.ac.in)}}%
\thanks{$^{2}$Twinkle Tripathy is with the Department of Electrical Engineering, Indian Institute of Technology Kanpur, India. (email: 
{\tt\small ttripathy@iitk.ac.in).}}%
\thanks{$^{3}$Rajasekhar Anguluri is with the Department of Computer Science and Electrical Engineering,  University of Maryland, Baltimore County, MD 85281, USA. (e-mail: {\tt\small rajangul@umbc.edu}).}
 }
\begin{document}
\maketitle
\thispagestyle{empty}
\pagestyle{empty}

\begin{abstract}
In this paper, we explore the property of eventual exponential positivity (EEP) in complex matrices. We show that this property holds for the real part of the matrix exponential for a certain class of complex matrices. Next, we present the relation between the spectral properties of the Laplacian matrix of an unsigned digraph with complex edge-weights and the property of real EEP. Finally, we show that the Laplacian flow system of a network is stable when the negated Laplacian admits real EEP. Numerical examples are presented to demonstrate the results.
\end{abstract}

\section{INTRODUCTION}
\looseness=-1 Consensus problems have been studied widely in the literature of multi-agent systems. Applications stretch from the synchronization of coupled oscillators in power and brain networks to the coordination among autonomous vehicles. For Laplacian flow systems\footnote{The Laplacian flow is a linear autonomous system evolving in continuous time, where the system matrix is defined by the network's Laplacian matrix.}, contemporary research focuses on consensus problems in networks having real edge-weights. In this work, we consider networks with complex edge-weights (shortly, complex networks), motivated by network science research on emergent behaviors in systems with inherently complex-valued weights. Some applications of complex-valued networks include quantum information \cite{kubota2021quantum}, electrodynamics \cite{muranova2020electrical}, computational social science \cite{hoser2005social}, graph neural networks \cite{kobayashi2010nnl}.

Our primary goal is to introduce and characterize the \emph{real eventual exponential positivity} (EEP) property for (complex-valued) Laplacians associated with complex networks. Extending this property to complex-valued Laplacian matrices is non-trivial. At the high-level, it states that the real component of the complex-valued exponential matrix $(e^{-Lt})$ is positive for sufficiently large positive value of $t$, where $L$ is the $n$-dimensional complex Laplacian matrix. At first glance, it may seem that the exponential matrix should be positive for any $t\geq 0$. While this is true in some cases (e.g., when $L$ is complex symmetric), it does not always hold in general (e.g. when $L$ is weakly connected). Nonetheless, we demonstrate that complex Laplacians for certain classes of networks with sign restrictions admits the real EEP property, thereby enabling the associated flow systems to achieve consensus.

In this paper, the properties of complex Laplacians are inspired from the real-valued Laplacian matrices. When the underlying networks are unsigned, the row sums of the Laplacians are zero. Owing to this, the spectrum of $L$ in complex networks consist of either a simple or a semi-simple zero eigenvalue along with the remaining eigenvalues lying in the open right half of the complex plane (hereafter ORHP). In particular cases addressed in the results, they exhibit a simple zero eigenvalue. Our results are valid for a certain class of complex Laplacians. We summarize our contributions below: 

\begin{itemize}
\item We introduce the notion of \emph{real EEP} for complex-valued matrices which guarantees the eventual positivity of the real part of the matrix exponential. Additionally, we deploy the strong complex Perron-Frobenius property \cite{vargamatrix} to derive the generalized conditions that ensure the real EEP property in a class of complex matrices.
\item Then, we present the Laplacian flows for complex networks. Subsequently, we examine the sufficient conditions in undirected and directed graphs required for the convergence of real and imaginary parts in complex-valued Laplacian flow systems. We utilize the real EEP property of $-L$ which ensures consensus among all the agents.
\end{itemize}
{\textit{Related Research}}: 
In \cite{altafini2019investigating},\cite{fontan2021properties}, the authors show that the property of EEP of the Laplacian $L$ in signed (real)networks leads to (strongly) connectedness of the (di)graph. They prove the inter relation between EEP and the spectrum of Laplacian matrices for signed networks with real edge-weights. In a similar manner, but employing sophistical tools for complex-valued matrices, we show that $-L$ exhibits real EEP iff the (di)graph is (strongly) connected. We also comment on the irreducibility of digraphs consisting of complex edge-weights. Our analysis holds for all unsigned networks which are undirected or weight-balanced digraphs. Further, we utilize the propoerty of real EEP to achieve consensus in multi-agent systems.
\section{PRELIMINARIES}
\subsection{Matrix Algebra}
Let $\mathbb{C}^{n,n}$ denote the space of $n\times n$ matrices with complex entries. For $n=1$, the space is $\mathbb{C}$, and let ${\Re}(m)$ and ${\Im}(m)$ denote the real and imaginary parts of $m \in \mathbb{C}$. A matrix $M:=[m_{ij}] \in  \mathbb{C}^{n,n}$ is nonnegative
if ${\Re}(m_{ij})\geq 0$ and ${\Im}(m_{ij})\geq 0$ for all $i,j \in \left \{ 1,2,...,n \right \}$. A matrix $M \in \mathbb{C}^{n,n}$  is real nonnegative if $\Re(m_{ij})\geq0$ for all $m_{ij}$ 
(hereafter $\Re(M)\geq 0$). A matrix is real eventually nonnegative if $\Re(M^{k})\geq 0 $ for some integer $ k\geq k_{0}$. Positive matrices, with appropriate qualifiers, satisfy the above inequalities with strict inequality.
The conjugate and the normal transpose of matrix $M$ are $M^{H}$ and $M^T$. A matrix is positive semi-definite (PSD) if the eigenvalues of the matrix lie in closed right half of s-plane. The symmetric part of matrix is denoted by $M_{s}=(M+M^{H})/2$. A matrix is normal if $MM^{H}=M^{H}M$ holds \cite{horn2012}. The spectrum of $M$ is denoted by $\operatorname{spec}(M)=\left \{ \lambda_{1},\lambda_{2},...,\lambda_{n} \right \}$ where $\lambda_{i}$ is the $i$-th eigenvalue of $M$ suffices. 
The eigenvalue $\lambda_{i}$ is semi-simple if its algebraic multiplicity equals geometric multiplicity. A permutation matrix (hereafter $P$) is a square binary matrix having a single $1$ in each row and column, and all other entries are zeros.


The following definition of Perron-Frobenius (hereafter PF) property for complex-valued matrices plays a crucial role in our analysis. More details are in \cite{varga2012}. 

\begin{definition}\label{def2} A matrix $M \in \mathbb{C}^{n,n}$ has the strong complex Perron-Frobenius (PF) property if its dominant eigenvalue $\lambda_{1}$ is positive, simple, with $\lambda_{1}>\left|\lambda_{i}\right|, i \in \left \{2,...,n \right \}$. The eigenvector $\mathbf{x}$ associated with $\lambda_1$ is such that  $\Re(\mathbf{x})>\mathbf{0}$ and $\mathbf{x}$ is the right PF eigenvector \cite{varga2012}.
\end{definition}

\subsection{Graph theory}
\looseness=-1
Let $\mathcal{G}(A)=(\mathcal{V} ,\mathcal{E} ,A)$ be the weighted graph, where $\mathcal{V}$ is the vertex set and $\mathcal{E}$ is the edge set. The adjacency matrix is $A=[a_{ij}] \in  \mathbb{C}^{n,n}$. If $A$ is complex symmetric, then $\mathcal{G}(A)$ is undirected. If edges in $\mathcal{G}(A)$ have orientations, the graph is directed. An undirected graph is connected if every vertex can reach every other vertex. A digraph is weakly connected if its undirected counterpart is connected. Instead, a digraph is strongly connected, if there exists a directed path between any pair of vertices. A graph is unsigned if all its edge-weights are nonnegative. When the edge-weights of a graph have positive and negative signs, then the graph is said to be a signed graph. In digraphs, out-degree of a vertex is the sum of edge-weights of all the outgoing edges, where as in-degree is the sum of incoming edge-weights. Collectively for all nodes, the degree matrices are given by $D_{out}=diag(A\mathbb{1}_{n})$ and $D_{in}=diag(A^{T}\mathbb{1}_{n})$, where $\mathbb{1}_{n}$ is the all ones vector. A digraph is weight-balanced if $D_{out}=D_{in}$ (see \cite{bullo2018lectures}).  
The Laplacian matrix is $L=D_{out}-A$ and it is complex symmetric for undirected graphs. 
\section{MAIN RESULTS} \label{sec3}
Here, we consider the property of real EEP for complex-valued Laplacian matrices, and explore implications of EEP property to study consensus in multi-agent systems. Building on Def.~\ref{def2}, we introduce the notion of real eventually exponentially positive (EEP) property for complex-valued matrices. 
\begin{definition}\label{def3} A matrix $M \in \mathbb{C}^{n,n}$ satisfies real EEP if there exists $ t_{0} \in \mathbb{N}$ such that $\Re(e^{Mt})>0$ for all $t\geq t_{0}$.
 \end{definition}
 This property of real EEP ensures that the system trajectories retains the same sign as the initial state vector after a transient period. 
We recall that for a real-valued $n \times n$ matrix $M$, EEP is concerned for a trajectory emanating from an initial point and remains nonnegative for all time thereafter whereas real EEP comments the same for real part of complex entries in the matrix. It guarantees that for sufficiently large time $t \geq t_{0}$, the $n$-dimensional real state of the system
\begin{align*}
\mathbf{\dot{x}}(t)=M\mathbf{x}(t)    
\end{align*}
will be component-wise nonnegative (i.e., $x_i(t)\geq 0$) if the initial state vector $x(0)$ is nonnegative. This property is essential in understanding the stabilization and long-term behavior of systems. We finally utilize this property to comment on the stability of negated Laplacians in complex domain.

The (real)eventually exponential positive matrices are mostly irreducible in (complex)real domain. Reducible matrices may eventually turn exponentially nonnegative \cite{noutsos2008reachability}. Therefore, it is important to study irreducibility for complex networks as both irreducibility and real EEP relate to the strongly connected nature of graphs. Ultimately, we link these properties to achieve consensus in later sections (see Remark \ref{rem:4}). Hence, we address it as follows.
\subsection{Irreduciblity in directed networks}
Irreducibility is a property observed in matrices which is used to evaluate connectivity in digraphs. The property of irreducibility has been thoroughly studied for nonnegative matrices \cite{bullo2018lectures} and has been extended to signed and complex-valued matrices \cite{vargamatrix}. 
The definition along with a key result on irreducibility and its relation to digraphs are discussed next.
\begin{definition}\label{def5}
For $n\geq 2$, the matrix $M\in \mathbb{C}^{n, n}$ is said to be reducible if there exists an $n\times  n$ permutation matrix $P$ such that 
\begin{equation}
 PMP^{T}=\begin{bmatrix}
M_{11} & M_{12}\\ 
 0& M_{22} 
\end{bmatrix}
\end{equation}
where $M_{11}$ is an $r\times r$ submatrix and $M_{22}$ is an $(n-r)\times(n-r)$ submatrix where $1\leq r<n$. If no such permutation matrix exists, then $M$ is irreducible \cite{vargamatrix}. 
\end{definition}
 The following result applies to all unsigned and signed digraphs.
\begin{lemma} \label{Thm:1}
 Consider a digraph for $n\geq 2$, $L \in \mathbb{C}^{n,n}$, then $L$ is irreducible if and only if $\mathcal{G}(A)$ is strongly connected.
\end{lemma}
\begin{proof} We have from Def.~\ref{def5} that if $L$ is irreducible, then there does not exist any $P$ such that $PLP^{T}$ is block triangular. We partition the set of vertices such that there are two set of vertices,$\mathcal{V}_{1},\mathcal{V}_{2}$ such that $\mathcal{V}_{1}\cap \mathcal{V}_{2}=\left \{\emptyset  \right \}$. Then, $L_{11}$ is an $r\times r$ submatrix, $L_{22}$ is an $(n-r)\times (n-r)$ submatrix where $1\leq r<n$. Similarly, $L_{12},~L_{21}$ are submatrices consisting of edges existing between the set of vertices. So, we have:
\begin{align*}
L\not\equiv \begin{bmatrix}
L_{11} &L_{12} \\ 
0 & L_{22}
\end{bmatrix}
\end{align*}
where $L_{11},L_{12}$ and  $L_{22}$ are not null submatrices. We know that $L=D_{out}-A$ which implies that there exist a walk from every set of vertices to every other set of vertices in matrix $A$. Although reduciblity implies that are no directed edges from $\mathcal{V}_{2}$ to $\mathcal{V}_{1}$ set of vertices where $\mathcal{V}_{1}\cap \mathcal{V}_{2}=\left \{\emptyset  \right \}$. Hence, the graph is strongly connected. So, $A$ is irreducible \cite{vargamatrix}, \cite{meyer2023}, \cite{zaslavsky1999jordan} i.e., there does not exist any $P$ such that $PAP^{T}$ is block triangular. Conversely, we know that 
$L=D_{out}-A$.
Thus, $PLP^{T}$ can not be equivalent to the block triangular form as $D_{out}$ is a diagonal matrix. Hence, $L$ is irreducible. \end{proof}

It is interesting that both adjacency and Laplacian matrices are irreducible whenever $\mathcal{G}(A)$ is strongly connected, regardless of whether these matrices are real-valued or complex-valued. 
\subsection{Real Eventually Exponentially Positive Laplacians}
 The property of \emph{(real) EEP} is closely related to the property of \emph{(real) EP (Eventual Positivity)} in the (complex) real domain. We know that eventual positivity of real-valued matrices is only achieved if and only if the matrices exhibit strong PF property \cite{noutsos2006perron}. In the case of complex-valued matrices, we also know that real eventual positivity is achieved if and only if the matrices exhibit strong complex PF property \cite{varga2012}.  
The following section thoroughly explores the relation between real EP property and real EEP property.
We introduce a set of matrices $\mathcal{P}$ where every matrix $M \in \mathcal{P}$ satisfies the following properties: 
\begin{enumerate}
\item Strong complex PF property: This property ensures a simple and positive dominant eigenvalue and hence the matrix is non-nilpotent.
 \item $\Re(\mathbf{x})\geq\left | \Im(\mathbf{x}) \right |  , \Re(\mathbf{z})\geq\left | \Im(\mathbf{z}) \right |$ where $\mathbf{x}$ and $\mathbf{z}$ are the right and the left eigenvectors corresponding to the dominant eigenvalue, respectively. 
\end{enumerate}
Mathematically, 
\begin{equation}
\label{eq:Pmatrix}
\begin{split}
    \mathcal{P}:=\{M \in\mathbb{C}^{n, n}  
 \mid~\Re(\mathbf{x})\geq\left | \Im(\mathbf{x}) \right |  , \Re(\mathbf{z})\geq\left | \Im(\mathbf{z}) \right | 
  & \\ \text { and }M \text{ follows the strong complex PF property } \}.
\end{split}
\end{equation}
The PF property of matrices aids in illustrating the behavior of matrices when raised to higher powers. The following result demonstrates the implications and advantages of complex PF properties. 
\begin{theorem}\label{stronglemma}
Consider $M$ and $M^{H} \in \mathbb{C}^{n,n}$ such that $\Re(\mathbf{x})\geq\left | \Im(\mathbf{x}) \right |, \Re(\mathbf{z})\geq\left | \Im(\mathbf{z}) \right | $ where $\mathbf{x}$ and $\mathbf{z}$ are the dominant right and left eigenvectors of $M$, respectively. Then, the following statements are equivalent for some scalar $d\geq 0$:
\begin{enumerate}[label=(\roman*)]
\item $ (M+dI) \text{ and } (M+dI)^{H} \in \mathcal{P}$ ,
\item $(M+dI)$ is real eventually positive for all $k\geq k_{0}$,
\item $M$ is real EEP
\end{enumerate}  
where $\mathcal{P}$ is as defined in eqn. \eqref{eq:Pmatrix}.
\end{theorem}%
\begin{proof}  $(i) \implies (ii)$: Using the Jordan decomposition form of $M , M^{H}\in \mathcal{P}$ (set $ d=0 \text{ in } (M+dI) \text{ and } (M+dI)^{H} \in \mathcal{P}$), \begin{equation}\label{jordaneqn1}
M=\begin{bmatrix}
 \mathbf{x}~|~
X_{n,n-1}\end{bmatrix} \begin{bmatrix} \lambda_{1} & 0 \\ 0 & D_{n-1} \end{bmatrix} \begin{bmatrix}
\mathbf{y}^{^{H}}\\Y_{n-1,n} \end{bmatrix}
\end{equation}  
where $\mathbf{y}$ is the left eigenvector corresponding to the dominant eigenvalue of $\lambda_{1}$ and $D_{n-1}$ is the square matrix consisting of all other eigenvalues of matrix $M$ in Jordan form.

Similarly we can decompose $(M+dI)$ as
\begin{equation}\label{jordaneqn2}
(M+dI)=\begin{bmatrix}
 \mathbf{x}~|~
X_{n,n-1}\end{bmatrix} \begin{bmatrix} (\lambda_{1}+d) & 0 \\ 0 & D'_{n-1}\end{bmatrix} \begin{bmatrix}
\mathbf{y}^{^{H}}\\Y_{n-1,n} \end{bmatrix}
\end{equation} where $D'_{n-1}$ is the square matrix consisting of remaining eigenvalues as $\lambda_{i}+d$ 
and $\lambda_{i} \in$ spec($M$).
For a nonnegative integer $k\geq k_0$, from  eqn.~\eqref{jordaneqn2}, we get
\begin{equation*}
\frac{(M+dI)^{k}}{(\lambda_{1}+d)^{k}} =\left[\mathbf{x} \mid X_{n, n-1}\right]\left[\begin{array}{c|c}
1 & 0 \\
\hline 0 & \frac{1}{(\lambda_{1}+d)^{k}} (D'_{n-1})^{k}
\end{array}\right]\left[\begin{array}{c}
\mathbf{y}^{^{H}} \\
\hline Y_{n-1, n}
\end{array}\right]
\end{equation*}
After simple manipulations, we arrive at \begin{equation*}
\lim _{k \rightarrow \infty} \frac{(M+dI)^{k}}{(\lambda_{1}+d)^{k}}=\mathbf{x} \mathbf{y}^{H}=\alpha \mathbf{x}\mathbf{z}^{H}.
\end{equation*}
We know that, $\mathbf{y}^{H}\mathbf{x}=1$ as they are normalized eigenvectors. We can express $\mathbf{y}=\alpha \mathbf{z}$ where $\alpha \in \mathbb{R}_{>0}$ satisfies the assumption of positive real parts in the dominant eigenvectors. Hence, we conclude that 
\begin{equation}\label{real}
\Re\left ( \lim _{k \rightarrow \infty} \frac{(M+dI)^{k}}{(\lambda_{1}+d)^{k}}\right )=\Re(\alpha \mathbf{x}\mathbf{z}^{H}) > \mathbf{0}_{n, n}.    
\end{equation} 
  
 $(ii) \implies (i)$: Let $B=(M+dI)$. The scalar $d\geqslant 0$ can always be chosen such that $B$ is non-nilpotent (i.e., $B^{k}\neq O$ for a positive integer $k$). From Thm. 2.3 \cite{varga2012}, it can be proved that  $B$ and $B^{H}\in\mathcal{P}$. 

$(ii)\implies(iii):$ Let $(M+dI)$ be real eventually positive. Thus, $\Re((M+dI)^{k}) > 0$ for any nonnegative integer $k\geq k_{0}$.
Then, the matrix exponential can be written as:
\begin{align}\label{exp1}
e^{((M+dI)t)} &= I + (M+dI)t + \frac{(M+dI)^{2} t^2}{2!} + \ldots \frac{(M+dI)^{k_{0}} t^{k_{0}}}{k_{0}!}\nonumber \\
    &\quad  + \sum_{k=k_0+1}^{\infty}\frac{(M+dI)^{k}t^{k}}{k!}
\end{align}
The higher order terms in Eqn.~\eqref{exp1} (terms with index greater than $k_0$) are real positive by definition of real eventual positivity. On the other hand, there exists a time $t=t_{0}$ such that the sum of first $k_{0}$ terms of Eqn.~\eqref{exp1} are dominated by the sum of rest of the terms with powers $k \geq k_{0}$ for all $t \geqslant t_{0}$. 
Hence, the whole summation $e^{(M+dI)t}$ is positive and $(M+dI)$ is real EEP. Now, $M$ inherits the property of real EEP from $(M+dI)$. This follows directly from the fact that  $e^{(Mt)}$ can be expressed as $e^{-dt}e^{(M+dI)t}$ where $e^{-dt}$ is positive. Hence, $M$ is real EEP.
 
$(iii)\implies(ii):$ Because $M$ is real EEP, there exists $k=k_{0}$ such that $\Re((e^M)^k)>0$ for all $k \geq k_{0}$. 
We know from the proof of $(ii)\implies(i)$ that real eventual positivity is equivalent to strong complex PF property. Hence, the dominant eigenvalue of $e^M$ and $(e^M)^{H}$ is $e^{\lambda_{1}}$ which is real, simple and positive. It follows that $e^{\lambda_{1}}>e^{\mu_{i}}$, where $\mu_{i}$ represents the $i^{th}$ of the remaining $n-1$ eigenvalues of $M$. Since $\mu_{i} \in \mathbb{C}$, $e^{\lambda_{1}}>|e^{\mu_{i}}|=|e^{\Re(\mu_{i})+i\Im(\mu_{i})}|=e^{\Re(\mu_{i})}$. Consider the Jordan decomposition of matrix exponential of ${M}$ given similar as Eqn~\eqref{jordaneqn1}:
\begin{equation}\label{expjordan}
e^{M}=\begin{bmatrix}
 \mathbf{x}~|~ 
X_{n,n-1}\end{bmatrix} \begin{bmatrix} e^{\lambda_{1}} & 0 \\ 0 & e^{D_{n-1}} \end{bmatrix} \begin{bmatrix}
\mathbf{y}^{H}\\Y_{n-1,n} \end{bmatrix}
  \end{equation} 
It now follows that $\operatorname{spec}(M)=\{\lambda_1,~\mu_i, i\in\{1,2,..,n-1\}\}$ such that $M$ and $e^{M}$ have the same eigenspaces. Hence, $\lambda_{1}$ is the spectral abscissa of $M$. We can always choose $d\geqslant 0$ suitably such that $(\lambda_{1}+d) > |\mu_i + d| \geqslant 0$ holds for all $\mu_i \in \operatorname{spec}(M)$ implying that $(M+dI),(M+dI)^{H}$ satisfies the strong complex PF property. In conjunction with the condition that $\Re(\mathbf{x})\geq\left | \Im(\mathbf{x}) \right |, \Re(\mathbf{z})\geq\left | \Im(\mathbf{z}) \right |$, it implies that $(M+dI),(M+dI)^{H} \in \mathcal{P}$ and is equivalently real eventually positive. The proof is now complete. 
\end{proof}
The strong complex PF property allows us to study the real EEP in arbitrary complex-valued matrices. We specialize the property of real EEP to the complex-valued Laplacian matrices, focusing on the case where the Laplacian matrix is associated with an underlying multi-agent network. We begin with the undirected networks. 

\subsubsection{Undirected unsigned networks} 
In an undirected unsigned graph, the flow of information is bidirectional and edge-weights are nonnegative. Hence, the Laplacian matrix $L$ is complex symmetric and the right and left eigenvectors corresponding to `0' eigenvalue lie in $\operatorname{span}\left \{ \mathbb{1}_{n}\right \}$. Because our motivation is to study the qualitative properties (e.g., consensus) of Laplacian flow systems (see Thm.~\ref{Thm3}), we state our results for $-L$. 

\begin{theorem}\label{Thm:undirected}
Consider an unsigned undirected graph $\mathcal{G}(A)$ such that $L \in \mathbb{C}^{n,n}$. The following statements are equivalent:
\begin{enumerate}[label=(\roman*)]
\item $-L$ is real EEP,
\item $-L$ has `0' as a simple eigenvalue.
\end{enumerate}
\end{theorem}
\begin{proof}
$(i)\implies(ii)$: 
Define the translated matrix as,
\begin{align}
\label{eq:B}
B=dI_{n}-L    
\end{align}
where $d\in\mathbb{R}$. Note that $e^{B}=e^{dI_{n}-L}=e^{d} e^{-L}$, where $e^d$ is a positive scalar. From this identity, we see that $B$ is real EEP if $-L$ is real EEP. Furthermore, 
$B\mathbb{1}_{n}=d\mathbb{1}_{n}-L\mathbb{1}_{n}=d\mathbb{1}_{n}$.


For any undirected $\mathcal{G}(A)$, the matrix $-L$ has `0' as a simple or semi-simple eigenvalue, and the remaining eigenvalues lie in LHP \cite{bullo2018lectures}. 

Let us choose 
\begin{align}\label{d}
 d>\underset{i=2,..,n}{max} \frac{\left|\lambda_i(L)\right|}{2},   
\end{align}
and note that `d' is the dominant eigenvalue of $B$. It is also a simple eigenvalue because $B$ satisfies the strong complex PF property. Thus, the right and the left eigenvectors are in $\operatorname{span}\left \{ \mathbb{1}_{n}\right \}$; and hence, from Thm.~\ref{stronglemma}, it follows that $B,~B^{H} \in \mathcal{P}$. Finally, since $B$ has `d' as a simple eigenvalue, it forces $-L$ to have `0' as a simple eigenvalue.  
%

$(ii)\implies(i)$: Let $-L$ has `0' as a simple eigenvalue. Further, let `d' be given in Eqn.~\eqref{d}. Then, $B=dI-L$ has `d' as its dominant eigenvalue and the corresponding right eigenvector is in $\operatorname{span}\left \{\mathbb{1}_{n}\right \}$. Since $B$ is complex symmetric, both the dominant left and right eigenvectors coincide and are in $\operatorname{span}\left \{  \mathbb{1}_{n}\right \}$. Thus,
$B,~B^{H} \in \mathcal{P}$. From Thm.~\ref{stronglemma}, 
$-L$ is real EEP. 
The proof is complete.
\end{proof}
Observe that Thm.~\ref{Thm:undirected} does not explicitly assume that $\mathcal{G}(A)$ is connected. But condition Thm.~\ref{Thm:undirected}$(ii)$ implies this (\cite{bullo2018lectures}, refer Ch.6). Furthermore, $-L$ (which follows real EEP property) is equivalent to $L$ being positive semi-definite. 

Our previous results discuss the interplay between the eigenspectrum, the property of real EEP, and the positive semi-definiteness of $L$. Crucial to our analysis is the assumption that the graph $\mathcal{G}(A)$ is undirected which ensures that $L$ is (complex) symmetric. In the rest of the section, we relax the assumption that the graph is undirected and thereby weaken the symmetry assumption. 
\subsubsection{Directed unsigned networks}
Consider a strongly connected digraph $\mathcal{G}(A)$ having no self-loops. For the digraph, the adjacency and Laplacian matrices may or might not be complex symmetric. Symmetry implies the graphs are weight-balanced but the vice versa is not true in general. Weight-balanced graphs allow us to invoke the strong complex PF property for the translated matrix $B$ in Eqn.\eqref{eq:B}, and ultimately, derive a result like Thm.~\ref{Thm:undirected} for digraphs. 
\begin{theorem}\label{Thm:2}
Consider an unsigned digraph $\mathcal{G}(A)$ such that $L\in\mathbb{C}^{n,n}$ corresponds to a strongly connected weight-balanced graph, then the following statements are equivalent:
\begin{enumerate}[label=(\roman*)]
    \item $-L$ is real EEP,
    \item $-L$ has `0' as a simple eigenvalue.
\end{enumerate}
\end{theorem}
\begin{proof}
$(i)\implies(ii)$: Since $\mathcal{G}(A)$ is weight-balanced and strongly connected, the right and left eigenvectors for eigenvalue `d' of $B$ are in  $\operatorname{span}\left \{  \mathbb{1}_{n}\right \}$. The remaining proof follows in the same lines as that of Thm. \ref{Thm:1}; where $B$ and `d' are defined in Eqns. \eqref{eq:B} and \eqref{d}, respectively.  Therefore, $-L$ has a (simple) `0' eigenvalue. 

$(ii)\implies(i)$: If $-L$ has `0' as a simple eigenvalue, then consider the translated matrix $B$ defined in Eqn. \eqref{eq:B} with `d' as given in Eqn.~\eqref{d}. Now, the goal is to ensure $B,B^{H} \in \mathcal{P}$ which implies $-L$ to be real EEP. From Eqn. \eqref{eq:B}, $B$ has `d' as its dominant eigenvalue. Also, the dominant left and right eigenvectors are in $\operatorname{span}\left \{  \mathbb{1}_{n}\right \}$ as $\mathcal{G}(A)$ is weight-balanced and strongly connected. Thus, $B,B^{H} \in \mathcal{P}$. Using Thm.~\ref{stronglemma}, this implies that $-L$ is real EEP. 
\end{proof}
\begin{remark}
The spectrum of $L$ consists of some eigenvalues at the origin and rest are in ORHP. Moreover for strongly connected graphs, the condition Thm. \ref{Thm:2}$(ii)$ straightforwardly implies that $L_{s}$ (the symmetric part of $L$) is positive semi-definite. 
\end{remark}
 Assertions in Thm.~\ref{Thm:2} are valid for only strongly connected and weight-balanced digraphs. Our future work will address if these assertions are valid for weakly connected graphs. 
Hereafter, we present the following numerical examples to illustrate the aforementioned results. 
\begin{example}\label{ex1}
Consider the undirected graph having complex edge-weights shown in Fig.~\ref{fig:undirected}
\begin{figure}[ht]
  \centering
  \begin{tabular}{p{0.21\textwidth}p{0.23\textwidth}}
    \centering
    \includegraphics[scale=0.53]{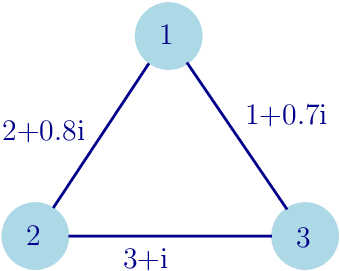}
    \caption{Undirected graph}
    \label{fig:undirected}
&
    \centering
    \includegraphics[scale=0.53]{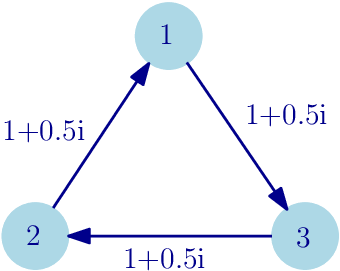}
    \caption{Digraph}
    \label{fig:digraph}
  \end{tabular}
  \label{fig:examples}
\end{figure} 
For the undirected graph, the Laplacian matrix is given by 
\begin{equation*}
L_{1}=\begin{bmatrix}
3+1.5i  & -2-0.8i &  -1-0.7i \\
 -2-0.8i &  5+1.8i & -3-i\\ 
-1-0.7i & -3-i & 4+1.7i\\
\end{bmatrix}. 
\end{equation*}
Here, the Laplacian matrix is complex symmetric. Unlike the real symmetric matrices, complex valued symmetric matrices may have complex eigenvalues. We will observe this in the spectrum of $L$ in this example.
The matrix exponential function of the negated Laplacian matrix evaluated at $t=1$ to determine the property of real EEP in Laplacian matrix as stated in the above results.  
\begin{equation*}
\left.\begin{matrix}
e^{-L_{1}t} 
\end{matrix}\right|_{t=1} = \renewcommand{\arraystretch}{1} 
\setlength{\arraycolsep}{2pt}    
\begin{bmatrix}
\scriptsize 0.33 + 0.01i & \scriptsize 0.33 + 0.002i & \scriptsize 0.34+0.005i\\ 
\scriptsize 0.33 +0.002i & \scriptsize 0.33-0.006i & \scriptsize 0.33 +0.001i\\ 
\scriptsize 0.34 + 0.005i & \scriptsize 0.33 -0.001i & \scriptsize 0.33 - 0.003i
\end{bmatrix}.
\end{equation*}
The eigenvalues are $\operatorname{spec}(L_{1})$=$\left \{ 0,~4.26+2.24i,~ 7.7+ 2.7i \right \}$. Set $d \geq 4.5$ and observe that $\operatorname{\Re}(B)$ is eventually positive, where $B=dI-L$. Consequently, $B \in \mathcal{P}$; so the `0' eigenvalue of $-L_{1}$ is simple, which is equivalent to $-L_1$ satisfying real EEP (see Thm. \ref{Thm:2}). Finally, the eigenvalues of $L_{1}=L_1^T$ lie in the RHP, and hence, $L_{1}$ is PSD. 
 \end{example}
\begin{example}\label{ex2}
 We consider a strongly connected digraph (see Fig. \ref{fig:digraph}) that is weight-balanced and yields a normal Laplacian.
 
The Laplacian matrix here is normal but it is neither symmetric nor Hermitian. Note that not every weight-balanced matrix is normal and the vice-versa also does not hold true.
 \begin{equation*}   
L_{2}=\begin{bmatrix}
1+0.5i  & 0 &  -1-0.5i \\
  -1-0.5i & 1+0.5i & 0\\ 
 0& -1-0.5i & 
1+0.5i\\
\end{bmatrix}\end{equation*}  
The matrix exponential function evaluated at $t=1$ is, 
\begin{equation*}
\left.\begin{matrix}
e^{-L_{2}t} 
\end{matrix}\right|_{t=1} = \renewcommand{\arraystretch}{1} 
\setlength{\arraycolsep}{2pt}    
\begin{bmatrix}
\scriptsize 0.38 + 0.11i & \scriptsize 0.21+ 0.1i & \scriptsize 0.42+ 0.01i\\ 
\scriptsize0.42+ 0.01i & \scriptsize 0.38-0.11i & \scriptsize 0.21 + 0.1i\\ 
\scriptsize 0.21+ 0.1i& \scriptsize0.42+ 0.01i & \scriptsize 0.38- 0.11i
\end{bmatrix}
\end{equation*}
The eigenvalues are $\operatorname{spec}(L_{2})$=$\left \{ 0, 1.06+1.61i, 1.93-0.11i \right \}$. By similar reasoning in Example.~\ref{ex1}, we conclude that $-L_2$ satisfies real EEP. Finally, because $L_2$ corresponds to a weight-balanced digraph and the eigenvalues are in the RHP, $L_{s}$ is PSD. 
\end{example}
The final result of this section below states that a Laplacian flow system (defined over complex linear multi-agent system) achieves consensus if the underlying complex-valued Laplacian satisfies the real EEP property. To our knowledge, this result is the first non-trivial extension of consensus for real-valued to complex-valued systems.
The following result applies to unsigned networks and weight-balanced digraphs where the negated Laplacian matrix exhibits the real EEP property. An example of power network is illustrated in Sec. \ref{sec5}, where the edges are signed, but since $-L$ possesses the real EEP property, consensus is achieved. This makes it clear that our focus is solely on the real EEP property to achieve consensus.
\begin{theorem}\label{Thm3}
 Consider a digraph $\mathcal{G}(A)$ such that $-L \in  \mathbb{C}^{n,n}$ is real EEP. Then, the Laplacian flow system, 
 \begin{align}
 \label{eq:Laplacian_flow}
    \dot{\mathbf{x}}(t) = -L \mathbf{x}(t),
 \end{align}
 attains consensus. 
 \end{theorem}
\begin{proof}
The solution of Laplacian flow \eqref{eq:Laplacian_flow} is given by
\begin{equation}\label{eqn5}
    \mathbf{x}(t)=e^{-Lt}\mathbf{x_{0}},
\end{equation}
where $\mathbf{x_{0}}$ is the vector of initial conditions. Provided that
\begin{equation*}
    \Re(e^{-Lt})> 0
\end{equation*} for $t\geqslant t_0$, it follows that $-L$ has a `0' simple eigenvalue and other eigenvalues lie in OLHP. Then, using the Jordan canonical decomposition, one can show that
$\lim _{t \rightarrow \infty} e^{-Lt}= \alpha\mathbb{1}_{n} \mathbf{z}^{H}$, where $\mathbb{1}_{n}$ and $\mathbf{z}$ are the right and left eigenvectors corresponding to the `0' eigenvalue, respectively. Finally, as $t\to\infty$, the final solution is given by
\begin{equation} \label{eqn6}
\mathbf{x}(t)= \mathbb{1}_{n} \mathbf{z}^{H}\mathbf{x_{0}}
\end{equation}
The scalar $\mathbf{z}^{H}\mathbf{x_{0}}$ is constant, and hence, all states are equal. Thus, consensus is achieved.
\end{proof} 
   

From the complex inner product given in Eqn. \eqref{eqn6}, we see that the real and imaginary parts of the state trajectories may converge to different values (see Section \ref{sec5}). 

\begin{remark}\label{rem:4}
The strongly connectedness of a digraph $\mathcal{G}(A)$ implies the property of real EEP of $-L$. Likewise, we know that $L$ is irreducible iff $\mathcal{G}(A)$ is strongly connected (from Thm.~\ref{Thm:1}). This conveys that in a strongly connected graph, agents will always reach consensus. This follows from the fact that $-L$ is real EEP and equivalently irreducible for almost all the matrices. \end{remark} 
\section{SIMULATION RESULTS} \label{sec5}
This section verifies our results in Section~\ref{sec3}. Specifically, we discuss some Laplacian flow systems which attain consensus. The proof as follows from Thm.\ref{Thm3}.
\subsection{Complex-valued Laplacian: Synthetic Networks}
We consider a Laplacian flow system on a network with three nodes. The initial states of the agents are $(6+2i, 2-1i, 4+0.7i)$. 

\begin{figure*}[ht]
  \centering
  \begin{subfigure}{0.31\textwidth}
    \centering
    \includegraphics[scale=0.41]{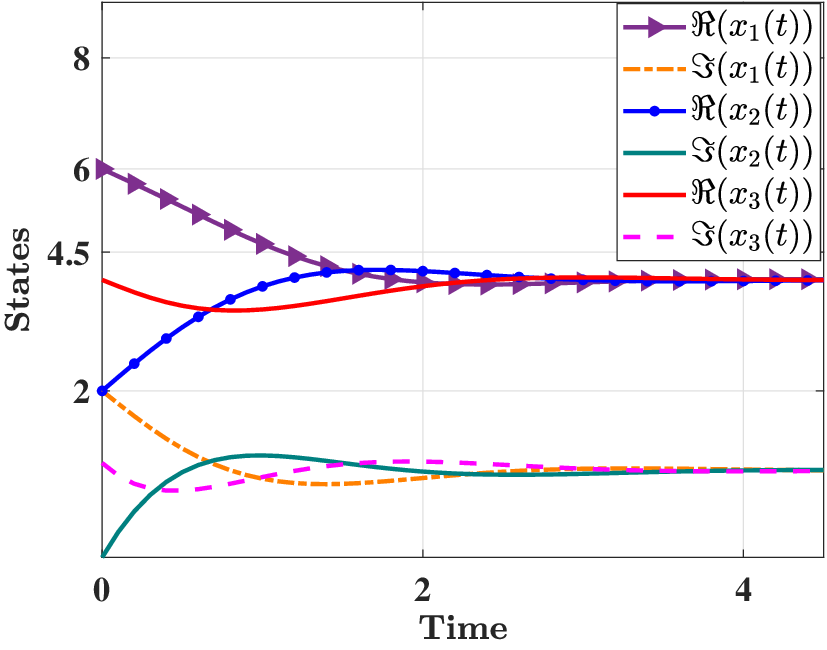}
    \caption{Convergence of states for a weight-balanced digraph.}
    \label{fig2}
  \end{subfigure}
  \hfill
  \begin{subfigure}{0.31\textwidth}
    \centering
    \includegraphics[scale=0.41]{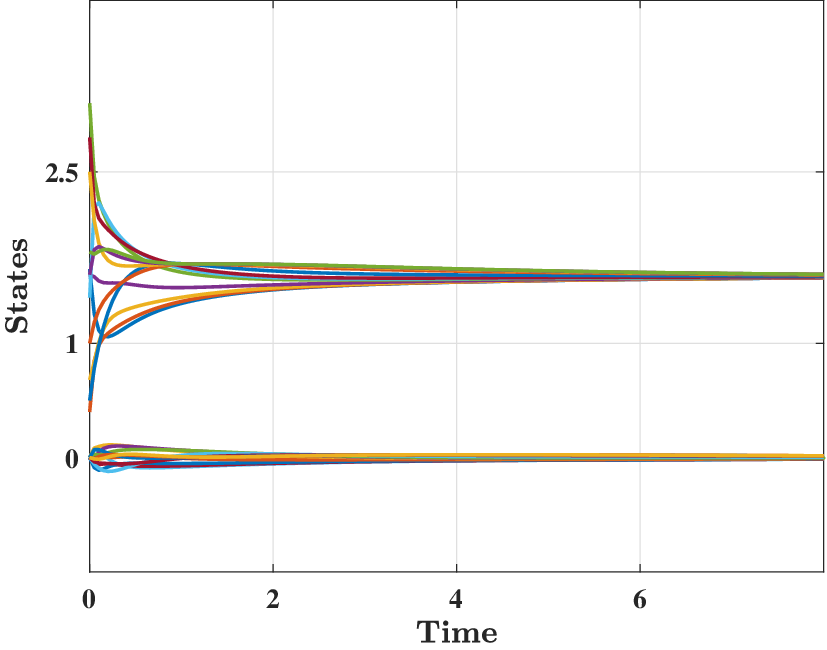}
    \caption{Convergence of states in a power distribution network.}
    \label{ybus}
  \end{subfigure}
  \hfill
  \begin{subfigure}{0.31\textwidth}
    \centering
    \includegraphics[scale=0.41]{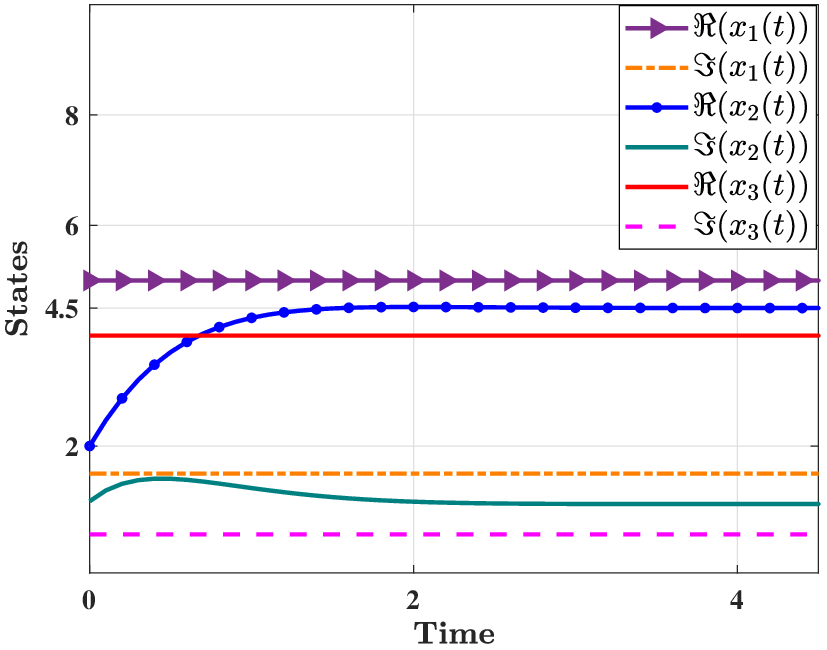}
    \caption{States failing to achieve consensus in a multi-agent network.}
    \label{fig:counter}
  \end{subfigure}
  \caption{Convergence properties of complex Laplacian flow systems}
  \label{fig:examplessim}
\end{figure*}

First, for a multi-agent network governed by the Laplacian matrix in Example~\ref{ex2}, the real and imaginary parts of the state trajectories are shown in Fig.~\ref{fig2}. Each trajectory converges (agents achieve average consensus). This behavior is expected because (i) the digraph is strongly connected and (ii) $-L$ has `0' as a simple eigenvalue. Both properties imply that $-L$ is real EEP; and consensus is guaranteed by Thm.~\ref{Thm3}. 
Since the digraph is weight-balanced, the agents achieve average consensus, i.e., the states converge to the mean of the initial conditions. The simulation result for Example \ref{ex1} converges similarly, given that the Laplacian matrix is complex symmetric.

\subsection{Complex-valued Laplacian: Power Network}
Here, we present an example of a Laplacian matrix that does not satisfy the non-negative adjacency assumption, but still, the state trajectories converge to a common value (see Fig.~ \ref{ybus}). The  Laplacian matrix here corresponds to the admittance matrix of power network. This is because of the real EEP property of $-L$. Additionally, we conclude that Thm. \ref{Thm:2} derives sufficient conditions for achieving the real EEP Property. We consider a complex-valued Laplacian matrix arising in a power distribution network with 12 nodes. We use the open source software Matpower \cite{zimmerman2010matpower} to obtain the Laplacian, which is complex symmetric but not Hermitian. 
\subsection{Counterexample}
We close this section with a modified version of Example \ref{ex2} where $-L$ does not satisfy the property of real EEP, and thus the state trajectories do not converge to a common value. The conditions derived for achieving real EEP in negated Laplacians and reaching consensus are sufficient conditions. Therefore, even if $-L$ is not real EEP, then agents may reach consensus. 
Consider the Laplacian matrix
 \begin{equation}\label{eq:counter}   
L_{3}=\begin{bmatrix}
0  & 0 &  0 \\
 -1-0.5i & 2+i & -1-0.5i\\ 
 0& 0& 0 \\
\end{bmatrix}
\end{equation}  
In this example, the digraph is weakly connected and thus the states do not converge (as shown in Fig.~\ref{fig:counter} to the same value. 
\section{CONCLUSIONS \& FUTURE WORK}
\looseness=-1
This paper introduces and investigates the property of real EEP in complex matrices. We derive several equivalent conditions to determine whether or not a complex-valued Laplacian matrix satisfies the real EEP. Under mild regularity conditions, we show that this property ensures consensus in the underlying unsigned Laplacian flow system, irrespective of whether the network is directed or undirected. Simulation results on synthetic and a power network system are presented to support our theoretical claims. 

Our results pave the way for many future research directions. Chief among them is to extend the results for signed graphs. We have seen that the weight-balanced digraphs are studied for now; work is currently in progress for the digraphs which are not weight-balanced. 

\addtolength{\textheight}{-12cm}   
       


\bibliography{References}
\bibliographystyle{unsrt}
\end{document}